\documentclass[fleqn,11pt,twoside]{article}

\usepackage{amsthm,amsthm,amssymb, color, xcolor,epsfig, graphics, subfigure}

\usepackage{amsmath, graphicx, latexsym, lscape }

\makeatletter
\newcommand{\copyrightnote}[2]{{\renewcommand{\thefootnote}{}
 \footnotetext{\small\it
\begin{flushleft}
 \copyright \ #1   #2  
\end{flushleft}}}}

\newcommand{\Name}[1]{\begin{flushleft}
                       \LARGE \bf #1
                       \end{flushleft}\vspace{-3mm}}

\newcommand{\Author}[1]{\begin{flushleft}
                       \it #1 \end{flushleft}}

\newcommand{\Address}[1]{\begin{flushleft}
                       \it #1 \end{flushleft}}

\newcommand{\Date}[1]{\begin{flushleft}
                      \small  \it #1 \end{flushleft}}

\newcommand{\evenhead}{Author \ name}
\newcommand{\oddhead}{Article \ name}

\renewcommand{\@evenhead}{
\hspace*{-3pt}\raisebox{-15pt}[\headheight][0pt]{\vbox{\hbox to \textwidth
{\thepage \hfil \evenhead}\vskip4pt \hrule}}}
\renewcommand{\@oddhead}{
\hspace*{-3pt}\raisebox{-15pt}[\headheight][0pt]{\vbox{\hbox to \textwidth
{\oddhead \hfil \thepage}\vskip4pt\hrule}}}
\renewcommand{\@evenfoot}{}
\renewcommand{\@oddfoot}{}

\long\def\@makecaption#1#2{%
  \vskip\abovecaptionskip
  \sbox\@tempboxa{\small \textbf{#1.}\ \ #2}%
  \ifdim \wd\@tempboxa >\hsize
    {\small \textbf{#1.}\ \ #2}\par
  \else
    \global \@minipagefalse
    \hb@xt@\hsize{\hfil\box\@tempboxa\hfil}%
  \fi
  \vskip\belowcaptionskip}

\newcommand{\JNMPnumberwithin}[3][\arabic]{%
  \@ifundefined{c@#2}{\@nocounterr{#2}}{%
    \@ifundefined{c@#3}{\@nocnterr{#3}}{%
      \@addtoreset{#2}{#3}%
      \@xp\xdef\csname the#2\endcsname{%
        \@xp\@nx\csname the#3\endcsname .\@nx#1{#2}}}}%
}

\renewenvironment{proof}[1][\proofname]{\par
  \normalfont
  \topsep6\p@\@plus6\p@ \trivlist
  \item[\hskip\labelsep\textbf{%
    #1\@addpunct{.}}]\ignorespaces
}{%
  \qed\endtrivlist
}

\newcommand{\resetfootnoterule} {
  \renewcommand\footnoterule{%
  \kern-3\p@
  \hrule\@width.4\columnwidth
  \kern2.6\p@}
}

\renewcommand{\footnoterule}{}

\makeatother

\theoremstyle{definition}
\newtheorem*{definition}{Definition}

\newtheorem{lemma}{Lemma}
\newtheorem{proposition}{Proposition}
\newtheorem{theorem}{Theorem}

\newtheorem{corollary}{Corollary}

\begin{document}

\renewcommand{\evenhead}{ {\LARGE\textcolor{blue!10!black!40!green}{{\sf \ \ \ ]ocnmp[}}}\strut\hfill 
Y. Alipour Fakhri}
\renewcommand{\oddhead}{ {\LARGE\textcolor{blue!10!black!40!green}{{\sf ]ocnmp[}}}\ \ \ \ \   
Finsler geometry in anisotropic superconductivity}

\thispagestyle{empty}
\newcommand{\FistPageHead}[3]{
\begin{flushleft}
\raisebox{8mm}[0pt][0pt]
{\footnotesize \sf
\parbox{150mm}{{Open Communications in Nonlinear Mathematical Physics}\ \  \ {\LARGE\textcolor{blue!10!black!40!green}{]ocnmp[}}
\ \ Vol.5 (2025) pp
#2\hfill {\sc #3}}}\vspace{-13mm}
\end{flushleft}}

\FistPageHead{1}{\pageref{firstpage}--\pageref{lastpage}}{ \ \ Article}

\strut\hfill

\strut\hfill

\copyrightnote{The author(s). Distributed under a Creative Commons Attribution 4.0 International License}

\Name{Finsler Geometry in Anisotropic Superconductivity: A Ginzburg--Landau Approach}

\Author{Y. Alipour Fakhri}

\Address{Faculty of Basic Sciences,
Department of Mathematics, Payame Noor University, Tehran, Iran.  E-mail: y\_ alipour@pnu.ac.ir}

\Date{Received October 23, 2025; Accepted October 28, 2025}

\setcounter{equation}{0}

\smallskip

\noindent
{\bf Citation format for this Article:}\newline
Y. Alipour Fakhri, Finsler geometry in anisotropic superconductivity: a Ginzburg-Landau approach,
{\it Open Commun. Nonlinear Math. Phys.}, {\bf 5}, ocnmp:16773, \pageref{firstpage}--\pageref{lastpage}, 2025.

\strut\hfill

\noindent
{\bf The permanent Digital Object Identifier (DOI) for this Article:}\newline
{\it 10.46298/ocnmp.16773}

\strut\hfill

\begin{abstract}
\noindent 
We present a rigorous generalization of the classical Ginzburg--Landau model to smooth, compact Finsler manifolds without boundary. This framework provides a natural analytic setting for describing anisotropic superconductivity within Finsler geometry. The model is constructed via the Finsler--Laplacian, defined through the Legendre transform associated with the fundamental function $F$, and by employing canonical Finsler measures such as the Busemann--Hausdorff and Holmes--Thompson volume forms. 
We introduce an anisotropic Ginzburg--Landau functional for complex scalar fields coupled to gauge potentials and establish the existence of minimizers in the appropriate Finsler--Sobolev spaces by the direct method in the calculus of variations. 
Furthermore, we analyze the asymptotic regime as the Ginzburg--Landau parameter $\varepsilon \to 0$ and prove a precise $\Gamma$--convergence result: the rescaled energies converge to the Finslerian length functional associated with the limiting vortex filaments. 
In particular, the limiting vortex energy is shown to equal $\pi$ times the Finslerian length of the corresponding current, thereby extending the classical Bethuel--Brezis--H\'elein result to anisotropic settings. 
These findings demonstrate that Finsler geometry unifies metric anisotropy and variational principles in gauge-field models, broadening the geometric scope of the Ginzburg--Landau theory beyond the Riemannian framework.
\end{abstract}

\label{firstpage}

\section{Introduction}\label{sec:Introduction}
The Ginzburg--Landau (GL) theory, introduced in the seminal work of Ginzburg and Landau \cite{GL1950}, has played a central role in mathematical physics and geometric analysis. In its classical form on a Euclidean or Riemannian manifold $(M,g)$, the model couples a complex order parameter to a gauge potential through a variational energy, (see, e.g., the monographs \cite{BBH1994,SandierSerfaty2007}) for a comprehensive
mathematical treatment including vortex structures, energy asymptotics, and compactness.

Recent developments in differential geometry have highlighted \emph{Finsler geometry} as a natural non quadratic extension of Riemannian structures, where the metric dependence on
directions is encoded by a strongly convex norm $F(x,\cdot)$ on each tangent space $T_xM$ (see \cite{BaoChernShen2000,Shen2001}). Analytic tools suitable for this setting such as the Finslerian gradient, divergence, Laplacian, and Sobolev spaces have been developed in, for instance, \cite{OhtaSturm2009,Barthelme2016}.
From the viewpoint of applications, anisotropy is intrinsic in layered superconductors and related media, suggesting that a GL-type theory on Finsler manifolds is a natural
framework for modeling direction-dependent phenomena.

\paragraph{Contributions.}
In this paper we formulate and analyze a Finslerian version of the GL model. Our main
contributions are as follows:
\begin{itemize}
\item[(i)] 
We define an anisotropic GL functional on a smooth Finsler manifold $(M,F)$ using the Finsler--Laplacian (via the Legendre transform) together with a canonical Finsler measure (Busemann--Hausdorff or Holmes--Thompson). The model couples complex scalar fields to $U(1)$-gauge potentials and is invariant under the natural gauge action.

\item[(ii)] 
We establish the \emph{existence of minimizers} in appropriate
Finsler--Sobolev spaces. The proof follows the direct method of the calculus of variations, relying on coercivity, weak lower semicontinuity induced by the convexity of $F^*(x,\cdot)^2$, compact embeddings on compact manifolds, and a Coulomb gauge
fixing based on a background Riemannian co-metric uniformly equivalent to $F^*$.

\item[(iii)] 
We investigate the asymptotic regime $\varepsilon\to 0$ and prove a $\Gamma$--convergence result, after the usual $|\log\varepsilon|$ rescaling, the energies converge to the Finslerian length of rectifiable $1$-currents representing
vortex filaments. The analysis adapts the ball construction and lower bound techniques of Jerrard--Sandier \cite{JerrardSandier2000} to the anisotropic setting, together
with compactness/rectifiability tools from geometric measure theory (see \cite{FedererFleming1960}) and the classical GL scheme in \cite{BBH1994,SandierSerfaty2007}.
\end{itemize}

\paragraph{Standing assumptions and notation.}

Throughout the paper, $(M,F)$ denotes a compact smooth Finsler manifold (without boundary, unless stated otherwise). We write $F^*(x,\cdot)$ for the co-metric on $T_x^*M$ induced by
the Legendre transform, and $d\mu_F$ for a fixed smooth Finsler measure; when not specified we adopt the Busemann--Hausdorff measure. A smooth Riemannian co--metric $\gamma^*$, uniformly equivalent to $F^*$, is used to formulate Hodge operators for the Maxwell term. All function spaces are the Finsler--Sobolev spaces $H^1_F$ built upon $F^*$ and $d\mu_F$, and the gauge is fixed to Coulomb form when needed.

\paragraph{Organization of the paper.}
Section \ref{sec:prelim} collects the necessary background on Finsler analysis. In Section \ref{sec:functional} we introduce the Finslerian GL functional and its basic properties, including gauge invariance and well-posedness on $H^1_F$. The existence of
minimizers is proved in Section \ref{sec:existence}. Section \ref{sec:gamma} is devoted to the asymptotic analysis as $\varepsilon\to 0$, culminating in the $\Gamma$--limit
characterization of vortex filaments via the Finsler length functional.

\section{Preliminaries on Finsler Geometry}
\label{sec:prelim}

In this section we collect the analytic and geometric tools used later. Throughout, $M$ denotes a smooth, connected, compact $n$--manifold without boundary. All statements below extend to manifolds with smooth boundary under standard trace
assumptions; see the remarks at the end of the section.


A smooth, strongly convex \emph{Finsler structure} on $M$ is a continuous function $F:TM\to[0,\infty)$ such that:
\begin{enumerate}
\item 
$F$ is $C^\infty$ on $TM\setminus\{0\}$,
  
\item 
$F(x,\lambda y)=\lambda F(x,y)$ for all $\lambda>0$,
\item 
for each $(x,y)\in TM\setminus\{0\}$, the \emph{fundamental tensor}
\begin{align*}
g_{ij}(x,y)\ :=\ \frac{1}{2}\,\frac{\partial^2\!\big(F(x,y)^2\big)}{\partial y^i\,\partial y^j}
\end{align*}
is positive definite.
\end{enumerate}
We write $F(x,\cdot)$ for the Minkowski norm on $T_xM$. The \emph{dual norm} $F^*:T^*M\to[0,\infty)$
is defined for $\xi\in T_x^*M$ by
\begin{align*}
F^*(x,\xi)\ :=\ \sup\{\, \xi(v)\ :\ v\in T_xM,\ F(x,v)\le 1\,\}.
\end{align*}
Define the \emph{Legendre map} $L:T M\setminus\{0\}\to T^* M\setminus\{0\}$ fiberwise by
\begin{equation}\label{eq:Legendre}
L_x(y)\ :=\ \partial_y\!\Big(\frac 12 F(x,y)^2\Big)
\ =\ g_{ij}(x,y)\,y^j\,dx^i.
\end{equation}

\begin{proposition}\label{prop:Legendre-diffeo}
For each $x\in M$, $L_x:T_xM\setminus\{0\}\to T_x^*M\setminus\{0\}$ is a $C^\infty$ diffeomorphism.
Its inverse is given by $L_x^{-1}(\xi)=\partial_\xi\!\big(\tfrac 12 F^*(x,\xi)^2\big)$.
Moreover, for all $y\neq 0$ and $\xi\neq 0$,
\begin{equation}\label{eq:F-Fstar-compat}
F^*\!\big(x,L_x(y)\big)\ =\ F(x,y),\qquad
F\!\big(x,L_x^{-1}(\xi)\big)\ =\ F^*(x,\xi),
\end{equation}
and the Fenchel--Young relation holds:
\begin{equation}\label{eq:Fenchel-Young}
\frac 12 F(x,y)^2 + \frac 12 F^*(x,\xi)^2 \ \ge\ \xi(y).
\end{equation}
Equality hold if and only if $\xi=L_x(y)$, equivalently, $y=L_x^{-1}(\xi)$.
\end{proposition}

\begin{proof}
Fix $x$ and set $\Phi(y):=\frac 12 F(x,y)^2$. By strong convexity, the Hessian $\partial_{y}^2\Phi(y)=g(x,y)$ is positive definite for $y\neq 0$. Hence $\nabla_y\Phi=L_x$ has everywhere invertible differential on $T_xM\setminus\{0\}$, by the inverse function theorem $L_x$ is a local diffeomorphism. Since $\Phi$ is strictly convex and superlinear, $L_x$ is injective
and proper; therefore it is a global diffeomorphism onto its image, which is $T_x^*M\setminus\{0\}$.
Define $\Psi(\xi):=\sup_{y\neq 0}\{\xi(y)-\Phi(y)\}$, the Legendre transform of $\Phi$, standard convex
duality gives $\Psi(\xi)=\frac 12 F^*(x,\xi)^2$ and $\nabla_\xi\Psi=L_x^{-1}$. The identities
\eqref{eq:F-Fstar-compat} and \eqref{eq:Fenchel-Young} are the usual equality cases in Fenchel duality, using strict convexity and $1$--homogeneity of $F$ and $F^*$.
\end{proof}

\begin{definition}\label{def:gradF}
For $u\in C^\infty(M)$, the \emph{Finsler gradient} $\nabla_F u(x)\in T_xM$ is defined by
\begin{equation}\label{eq:gradF}
du\ =\ L_x\big(\nabla_F u(x)\big).
\end{equation}
Equivalently, $\nabla_F u(x)=\partial_\xi\big(\frac 12 F^*(x,du_x)^2\big)$. By \eqref{eq:F-Fstar-compat} we have $F\!\big(x,\nabla_F u(x)\big)=F^*(x,du_x)$.
\end{definition}


Two canonical smooth measures on $(M,F)$ will be used:

\paragraph{Busemann--Hausdorff.}
For $x\in M$ let $B_F(x):=\{y\in T_xM:\ F(x,y)<1\}$ and $B^n\subset\mathbb{R}^n$ the Euclidean unit ball.
The Busemann--Hausdorff measure is
\begin{equation}\label{eq:BH}
d\mu_{BH}(x)\ :=\ \frac{\mathrm{vol}(B^n)}{\mathrm{vol}\big(B_F(x)\big)}\,dx^1\wedge\cdots\wedge dx^n.
\end{equation}

\paragraph{Holmes--Thompson.}
Let $B^*_F(x):=\{\xi\in T_x^*M:\ F^*(x,\xi)<1\}$. The Holmes--Thompson measure is
\begin{equation}\label{eq:HT}
d\mu_{HT}(x)\ :=\
\frac{1}{\mathrm{vol}(B^n)}\!\int_{B^*_F(x)}\! d\xi^1\cdots d\xi^n\,dx^1\wedge\cdots\wedge dx^n.
\end{equation}
In what follows we fix once and for all a smooth Finsler measure $d\mu_F$ chosen among $\{\mu_{BH},\mu_{HT}\}$ and write
\begin{equation}\label{eq:mu-density}
d\mu_F\ =\ \sigma(x)\,dx^1\wedge\cdots\wedge dx^n,\qquad \sigma\in C^\infty(M),\ \sigma>0.
\end{equation}

\begin{definition}\label{def:div-muF}
For a $C^1$ vector field $X=X^i\partial_i$ define
\begin{equation}\label{eq:divergence}
\mathrm{div}_{\mu_F} X\ :=\ \frac{1}{\sigma(x)}\,\partial_i\!\big(\sigma(x)\,X^i\big).
\end{equation}
\end{definition}

\begin{lemma}\label{lem:ibp}
For $u\in C^\infty(M)$ and $X\in C^1(TM)$,
\begin{equation}\label{eq:ibp}
\int_M \! du(X)\,d\mu_F\ =\ -\int_M \! u\ \mathrm{div}_{\mu_F}X\,d\mu_F.
\end{equation}
\end{lemma}

\begin{definition}\label{def:laplacian}
For $u\in C^\infty(M)$, the \emph{Finsler Laplacian} associated with $d\mu_F$ is
\begin{equation}\label{eq:laplacian}
\Delta^F_{\mu_F} u\ :=\ \mathrm{div}_{\mu_F}\big(\nabla_F u\big).
\end{equation}
In local coordinates, combining \eqref{eq:divergence} and \eqref{eq:gradF} yields
\begin{equation}\label{eq:laplacian-local}
   \Delta^F_{\mu_F} u\ =\ \frac{1}{\sigma(x)}\,
   \partial_i\!\Big(\sigma(x)\, g^{ij}\!\big(x,\nabla_F u(x)\big)\,\partial_j u\Big),
\end{equation}
where $g^{ij}(x,\cdot)$ denotes the inverse matrix of $g_{ij}(x,\cdot)$ evaluated at $y=\nabla_F u(x)$,
i.e. via the Legendre correspondence $du=L_x(y)$.
\end{definition}

\begin{definition}\label{def:dirichlet}
The Finsler Dirichlet energy of $u\in C^\infty(M)$ is
\begin{equation}\label{eq:dirichlet-energy}
\mathcal{E}_F[u]\ :=\ \frac 12\int_M F^*\!\big(x,du\big)^2\, d\mu_F.
\end{equation}
\end{definition}

\begin{theorem}\label{thm:EL-Dirichlet}
For $u\in C^\infty(M)$ and $\varphi\in C^\infty(M)$,
\begin{align*}
\frac{d}{dt}\,\mathcal{E}_F[u+t\varphi]\Big|_{t=0}
\ =\ -\int_M \varphi\, \Delta^F_{\mu_F}u\, d\mu_F.
\end{align*}
Equivalently, the critical points of $\mathcal{E}_F$ are the (weak) solutions of
$\Delta^F_{\mu_F} u=0$.
\end{theorem}

\begin{proof}
Set $\xi_t:=d(u+t\varphi)=du+t\,d\varphi$. Using Proposition \ref{prop:Legendre-diffeo},
\begin{align*}
\frac{d}{dt}\,\frac 12 F^*(x,\xi_t)^2\Big|_{t=0}
\ =\ \Big\langle \partial_\xi\!\Big(\tfrac 12 F^*(x,\xi)^2\Big)\Big|_{\xi=du}\!,\, d\varphi\Big\rangle
\ =\ \langle \nabla_F u,\ d\varphi\rangle,
\end{align*}
where the last pairing is the natural one $T_xM\times T_x^*M\to\mathbb{R}$. Hence
\begin{align*}
\frac{d}{dt}\,\mathcal{E}_F[u+t\varphi]\Big|_{t=0}
\ =\ \int_M \! \langle \nabla_F u,\ d\varphi\rangle\, d\mu_F
\ =\ \int_M \! d\varphi(\nabla_F u)\, d\mu_F.
\end{align*}
Now apply Lemma \ref{lem:ibp} with $X=\nabla_F u$ to obtain
\begin{align*}
\int_M d\varphi(\nabla_F u)\, d\mu_F
\ =\ -\int_M \varphi\, \mathrm{div}_{\mu_F}(\nabla_F u)\, d\mu_F
\ =\ -\int_M \varphi\, \Delta^F_{\mu_F}u\, d\mu_F.
\end{align*}
\end{proof}


Fix any smooth Riemannian metric $\gamma$ on $M$ and denote by $|\cdot|_{\gamma^*}$ the norm on $T^*M$ induced by its co-metric $\gamma^*$. On compact $M$, the norms $F^*(x,\cdot)$ and $|\cdot|_{\gamma^*}$ are uniformly equivalent:

\begin{lemma}\label{lem:equivalence}
There exist constants $0<c_1\le c_2<\infty$ such that for all $(x,\xi)\in T^*M$,
\begin{equation}\label{eq:equiv}
   c_1\,|\xi|_{\gamma^*}\ \le\ F^*(x,\xi)\ \le\ c_2\,|\xi|_{\gamma^*}.
\end{equation}
\end{lemma}

\begin{proof}
Let $S:=\{(x,\xi)\in T^*M:\ |\xi|_{\gamma^*}=1\}$, which is compact. The map $(x,\xi)\mapsto F^*(x,\xi)$ is continuous and positive on $S$. Set $c_1:=\min_S F^*$ and $c_2:=\max_S F^*$. Then $0<c_1\le c_2<\infty$, and by homogeneity of $F^*$ the inequality \eqref{eq:equiv} follows for arbitrary $\xi$.
\end{proof}

\begin{definition}\label{def:sobolev}
For $1\le p<\infty$, define $W^{1,p}_F(M)$ as the completion of $C^\infty(M)$ with respect to
\begin{align*}
\|u\|_{W^{1,p}_F}^p\ :=\ \int_M |u|^p\, d\mu_F\ +\ \int_M F^*(x,du)^p\, d\mu_F.
\end{align*}
We write $H^1_F(M):=W^{1,2}_F(M)$. For complex--valued maps, set
\begin{align*}
\|u\|_{F^*}^2:=F^*(x,\mathrm{Re}\,du)^2+F^*(x,\mathrm{Im}\,du)^2,
\end{align*}
and define $H^1_F(M;C)$ analogously.
For $1$--forms, $H^1_F(M;T^*M)$ is defined using any fixed co-metric $\gamma^*$ and the equivalent norm
\begin{align*}
\|\alpha\|_{H^1}^2=\int_M \big(|\alpha|_{\gamma^*}^2+|\nabla^\gamma\alpha|_{\gamma^*}^2\big)\,d\mu_F,
\end{align*}
which is equivalent to any other choice by Lemma \ref{lem:equivalence}.
\end{definition}

\begin{proposition}\label{prop:poincare-rellich}
The space $H^1_F(M)$ is a Hilbert space.
Moreover, there exists $C>0$ such that for all $u\in H^1_F(M)$ with mean zero $\int_M u\,d\mu_F=0$,
\begin{equation}\label{eq:Poincare}
   \|u\|_{L^2(M,d\mu_F)}\ \le\ C\, \|F^*(x,du)\|_{L^2(M,d\mu_F)}.
\end{equation}
Finally, the embedding $H^1_F(M)\hookrightarrow L^2(M,d\mu_F)$ is compact.
\end{proposition}

\begin{proof}
By Lemma \ref{lem:equivalence} and \eqref{eq:mu-density}, the $H^1_F$--norm is equivalent to the standard
$H^1(M,\gamma)$ norm with respect to the smooth positive density $\sigma\,dx$:
\begin{align*}
\|u\|_{H^1_F}^2\ \simeq\ \int_M \big( |u|^2 + |\nabla^\gamma u|_\gamma^2\big)\,\sigma\,dx.
\end{align*}
Therefore $H^1_F(M)$ is isomorphic as a Hilbert space to $H^1(M,\gamma)$. The Poincar\'e inequality \eqref{eq:Poincare} follows from the usual Poincar\'e inequality for $(M,\gamma)$
(since $M$ is compact) and the norm equivalence, similarly for Rellich's compact embedding.
\end{proof}


For $1$--forms we will use a fixed smooth Riemannian co-metric $\gamma^*$ uniformly equivalent to $F^*$
(Lemma \ref{lem:equivalence}) to formulate Hodge operators for the Maxwell term. Write $\sharp_\gamma:T^*M\to TM$ for the $\gamma$--musical isomorphism and define
\begin{align*}
  d^\dagger_{\gamma,\mu_F}\eta\ :=\ -\,\mathrm{div}_{\mu_F}\!\big(\eta^{\sharp_\gamma}\big)
  \qquad(\eta\in\Omega^1(M)).
\end{align*}
Then for all $\phi\in C^\infty(M)$,
\begin{equation}\label{eq:adjoint}
   \int_M \langle d\phi,\eta\rangle_{\gamma^*}\, d\mu_F\ =\ \int_M \phi\, d^\dagger_{\gamma,\mu_F}\eta\, d\mu_F,
\end{equation}
which follows by Lemma \ref{lem:ibp} and the definition of $\mathrm{div}_{\mu_F}$. In particular, $d^\dagger_{\gamma,\mu_F}$ is the $L^2(d\mu_F)$--adjoint of $d$ acting on $0$--forms.

If $\partial M\neq\emptyset$, \eqref{eq:ibp} gains a boundary term
$\int_{\partial M}\! u\, \iota_\nu X\, d\sigma_F$, where $\nu$ is the outward conormal and $d\sigma_F$ the induced Finsler boundary measure; Dirichlet ($u=0$) or Neumann ($\iota_\nu X=0$) conditions recover \eqref{eq:ibp}. On noncompact manifolds, the results above hold under uniform bounds ensuring \eqref{eq:equiv} and a global Poincar\'e inequality (e.g. positive injectivity radius and bounded geometry with respect to some $\gamma$).

\section{The Finslerian Ginzburg--Landau Functional}
\label{sec:functional}

Let $(M,F)$ be a compact smooth Finsler manifold endowed with a fixed smooth Finsler measure $d\mu_F$ (either Busemann--Hausdorff or Holmes--Thompson) and a smooth Riemannian co--metric $\gamma^*$ uniformly equivalent to $F^*$
(cf.\ Lemma \ref{lem:equivalence}). For a complex scalar field $\psi:M\to\mathbb{C}$ and a real $1$-form $A\in\Omega^1(M)$ we set
\begin{align*}
D_A\psi := (d - iA)\psi \in \Omega^1(M;\mathbb{C}),
\end{align*}
and extend the Finsler co-norm to complex-valued $1$-forms by
\begin{align*}
\| \eta \|_{F^*}^2 := F^*(x,\Re \eta)^2 + F^*(x,\Im \eta)^2, \qquad \eta\in\Omega^1(M;\mathbb{C}).
\end{align*}
\begin{definition}
The \emph{Finslerian Ginzburg--Landau functional} is
\begin{equation}\label{eq:GL-functional-Finsler}
G_F[\psi,A]
:= \int_M \Big( \tfrac 12 \|D_A\psi\|_{F^*}^2
+ \tfrac{1}{2\lambda}\, \| dA \|_{\gamma^*}^2
+ \tfrac{1}{4\varepsilon^2}\,(1-|\psi|^2)^2 \Big)\, d\mu_F,
\end{equation}
for fixed parameters $\lambda>0$ and $\varepsilon>0$.
\end{definition}


\begin{proposition}\label{prop:gauge-invariance}
For every $\chi\in H^1(M;\mathbb{C})$,
\begin{align*}
(\psi,A)\mapsto (e^{i\chi}\psi,\, A + d\chi)
\end{align*}
leaves $G_F$ invariant. In particular, $G_F$ depends only on the gauge class of $(\psi,A)$.
\end{proposition}

\begin{proof}
Observe that $D_{A+d\chi}(e^{i\chi}\psi) = e^{i\chi} D_A\psi$.  
Since the norm $\|\,\cdot\,\|_{F^*}$ is rotation-invariant in the $(\Re,\Im)$-plane by definition, it follows that  
\begin{align*}
\|D_{A+d\chi}(e^{i\chi}\psi)\|_{F^*} = \|D_A\psi\|_{F^*}.
\end{align*}
Moreover, the Maxwell term depends only on $dA$, so $\|d(A+d\chi)\|_{\gamma^*} = \|dA\|_{\gamma^*}$.  
Finally, the potential term depends only on $|\psi|$.  
Substituting these observations into \eqref{eq:GL-functional-Finsler} yields the desired result.
\end{proof}

\begin{proposition}\label{prop:welldef}
If $\psi\in H^1_F(M;\mathbb{C})$ and $A\in H^1_F(M;T^*M)$, then $G_F[\psi,A]\in[0,\infty)$
and all terms in \eqref{eq:GL-functional-Finsler} are finite. Moreover $G_F$ is $C^1$ on $H^1_F(M;\mathbb{C})\times H^1_F(M;T^*M)$.
\end{proposition}

\begin{proof}
By Lemma \ref{lem:equivalence} and compactness of $M$,
$F^*(x,\cdot)$ is uniformly equivalent to $|\cdot|_{\gamma^*}$, hence $\|D_A\psi\|_{F^*}\in L^2(M,d\mu_F)$ when $\psi,A\in H^1_F$. The Maxwell term is in $L^1$ because $dA\in L^2$ and $d\mu_F$ is smooth. The potential term is in $L^1$
since $H^1_F\hookrightarrow L^4$ on compact $M$ (Sobolev and norm equivalence). $C^1$-regularity follows from the chain rule and smoothness/convexity of $F^*(x,\cdot)^2$,
plus bilinearity of $(\psi,A)\mapsto D_A\psi$.
\end{proof}


We compute the G\^ateaux derivative of $G_F$ at $(\psi, A)$ in the directions 
$(\varphi, B) \in H^1_F(M;\mathbb{C}) \times H^1_F(M;T^*M)$. 
By Definition~\ref{def:gradF}, Theorem \ref{thm:EL-Dirichlet} (with complex realification), and the adjoint operator $d^\dagger_{\gamma,\mu_F}$, it follows that:
\begin{proposition}[\textbf{First variation}]
\label{prop:EL}
For every smooth $(\psi,A)$ and test pair $(\varphi,B)$,
\begin{align*}
  \frac{d}{dt}\, G_F[\psi+t\varphi, A+tB]\Big|_{t=0}
  &= \Re \int_M \!\!\big\langle \nabla_F \psi,\ D_A\varphi - i B\,\psi \big\rangle \, d\mu_F \\
  &\quad\ + \frac{1}{\lambda}\int_M \!\!\langle dA,\ dB\rangle_{\gamma^*}\, d\mu_F
  - \frac{1}{2\varepsilon^2}\int_M \!(1-|\psi|^2)\,\Re(\overline{\psi}\,\varphi)\, d\mu_F.
\end{align*}
Equivalently, the critical points $(\psi,A)$ satisfy, in weak form,
\begin{equation}\label{eq:EL-weak}
\begin{cases}
 D_A^* D_A \psi \ =\ \dfrac{1}{2\varepsilon^2}\,(1-|\psi|^2)\,\psi,\\[6pt]
 d^\dagger_{\gamma,\mu_F}\, dA \ =\ \lambda\, \Im\big(\overline{\psi}\, D_A\psi\big),
\end{cases}
\end{equation}
where $D_A^*$ is the $L^2(d\mu_F)$-adjoint of $D_A$ induced by $F^*$.
\end{proposition}

\begin{proof}
We differentiate each term of the functional separately. 
For the kinetic term, we use the realification 
\begin{align*}
\|D_A\psi\|_{F^*}^2 = F^*(x, \Re D_A\psi)^2 + F^*(x, \Im D_A\psi)^2,
\end{align*}
and apply Theorem~\ref{thm:EL-Dirichlet} componentwise, noting that
\begin{align*}
\partial_t D_{A+tB}(\psi+t\varphi)\big|_{t=0} = D_A\varphi - i\,B\,\psi.
\end{align*}
For the Maxwell term, we integrate by parts using the adjoint operator $d^\dagger_{\gamma,\mu_F}$ 
(see Equation\eqref{eq:adjoint}). 
For the potential term, we differentiate the polynomial nonlinearity directly. Collecting the resulting terms with respect to the test functions $(\varphi,B)$ yields Equation \eqref{eq:EL-weak}.
\end{proof}

Suppose that
\begin{align*}
  \mathcal{C}\ :=\ \{\, A\in H^1_F(M;T^*M)\ :\ d^\dagger_{\gamma,\mu_F}A=0 \,\}
\end{align*}
is the Coulomb slice. A standard Hodge decomposition with respect to $\gamma$ ensures every gauge class contains a representative in $\mathcal{C}$.

\begin{lemma}\label{lem:lsc}
If $(\psi_k,A_k)\rightharpoonup (\psi,A)$ weakly in $H^1_F(M;\mathbb{C})\times H^1_F(M;T^*M)$
and strongly in $L^2$, then
\begin{align*}
G_F[\psi,A]\ \le\ \liminf_{k\to\infty} G_F[\psi_k,A_k].
\end{align*}
\end{lemma}

\begin{proof}
The map $\eta\mapsto \tfrac 12 \|\eta\|_{F^*}^2$ is convex in $\eta$ (sum of convex maps
$\xi\mapsto \tfrac 12 F^*(x,\xi)^2$ on real and imaginary parts), hence weakly lower
semicontinuous. The Maxwell term is quadratic and thus weakly l.s.c. The potential term is continuous under $L^2$ convergence by dominated convergence on compact $M$.
\end{proof}

\begin{lemma}\label{lem:coercive-slice}
There exist constants $C_1,C_2>0$ (depending only on $M,F,\gamma^*,d\mu_F,\lambda$) such that
for all $(\psi,A)\in H^1_F(M;\mathbb{C})\times\mathcal{C}$,
\begin{align*}
G_F[\psi,A]\ \ge\ C_1\Big( \|\psi\|_{H^1_F}^2 + \|A\|_{H^1_F}^2 \Big) - C_2.
\end{align*}
\end{lemma}

\begin{proof}
By Lemma \ref{lem:equivalence} and  Proposition \ref{prop:poincare-rellich}, $\|dA\|_{L^2}$ controls $\|A\|_{H^1}$ on $\mathcal{C}$ (standard elliptic estimate for the operator $d^\dagger d$ with Coulomb constraint). The kinetic term controls $\|\psi\|_{H^1_F}$ up to a constant via the diamagnetic inequality in the next lemma (and the potential term bounds $\|\psi\|_{L^4}$). Collect the bounds and absorb constants.
\end{proof}

\begin{lemma}\label{lem:diamagnetic}
For every $\psi\in H^1_F(M;\mathbb{C})$ and $A\in L^2(M;T^*M)$,
\begin{align*}
 F^*(x, d|\psi|)\ \le\ \|D_A\psi\|_{F^*}\qquad \text{a.e.\ on }M.
\end{align*}
Consequently,
\begin{align*}
\|\;|\psi|\;\|_{H^1_F}\ \lesssim\ \|D_A\psi\|_{L^2(d\mu_F)} + \|\psi\|_{L^2(d\mu_F)}.  
\end{align*}
\end{lemma}

\begin{proof}
Where $\psi\neq 0$,
$d|\psi|= \Re\!\big( \overline{\psi}/|\psi|\, D_A\psi \big)$
(pointwise identity). By the definition of $\|\cdot\|_{F^*}$ on complex $1$-forms and the triangle inequality for $F^*$ on real forms, $F^*(x,d|\psi|)\le \|D_A\psi\|_{F^*}$.
Extend by continuity across $\{\psi=0\}$ using an approximation and the fact that both sides belong to $L^2$. The $H^1_F$ estimate follows by integrating and adding $\|\psi\|_{L^2}$.
\end{proof}

All results in this section are stable under replacing $d\mu_F$ by any smooth positive density equivalent to it, and $\gamma^*$ by any co-metric uniformly equivalent to $F^*$.
Coercivity constants change by multiplicative factors but the functional framework and the Euler--Lagrange system remain the same. The $\Gamma$--limit in Section \ref{sec:gamma} will
be seen to be independent of these choices as well.

\section{Existence of Minimizers}
\label{sec:existence}
We now establish the existence of minimizers for the Finslerian
Ginzburg--Landau functional introduced in Section \ref{sec:functional}. The proof is entirely variational and relies on the geometric analysis developed in Section \ref{sec:prelim}.


Let
\begin{align*}
\mathcal{H}_F := H^1_F(M;\mathbb{C}) \times H^1_F(M;T^*M)
\end{align*}
be the natural energy space. Because $G_F$ is gauge invariant
(Proposition~\ref{prop:gauge-invariance}), we restrict to a fixed Coulomb slice
\begin{align*}
\mathcal{A}_C := \{\, (\psi,A)\in \mathcal{H}_F
       : d^\dagger_{\gamma,\mu_F} A = 0 \,\}.
\end{align*}
By Hodge decomposition on $(M,\gamma)$, every gauge class
contains a representative in $\mathcal{A}_C$; moreover, within $\mathcal{A}_C$ the gauge freedom reduces to the compact torus of harmonic forms, which does not affect coercivity or compactness.

\begin{lemma}\label{lem:gauge-fix}
For any $(\psi,A)\in \mathcal{H}_F$ there exists a unique
$\chi\in H^2(M;\mathbb{R})$ with mean zero such that
$(e^{i\chi}\psi, A+d\chi)\in\mathcal{A}_C$. Moreover, the map
$(\psi,A)\mapsto(e^{i\chi}\psi, A+d\chi)$ is continuous on $\mathcal{H}_F$.
\end{lemma}

\begin{proof}
Since $d^\dagger_{\gamma,\mu_F}(A+d\chi)
 = d^\dagger_{\gamma,\mu_F}A + \Delta^\gamma_{\mu_F}\chi$,
where $\Delta^\gamma_{\mu_F} = d^\dagger_{\gamma,\mu_F} d$
is a uniformly elliptic self-adjoint operator on $H^2(M)$,
there exists a unique solution $\chi$ with zero mean.
Continuity follows from elliptic estimates.
\end{proof}

Hence it suffices to minimize $G_F$ over $\mathcal{A}_C$.
\begin{lemma}\label{lem:exist-coercive}
There exist constants $C_1,C_2>0$ such that for all $(\psi,A)\in\mathcal{A}_C$,
\begin{equation}\label{eq:coercive-bound}
  G_F[\psi,A] \ge
  C_1\big( \|\psi\|_{H^1_F}^2 + \|A\|_{H^1_F}^2 \big) - C_2.
\end{equation}
\end{lemma}

\begin{proof}
Combine Lemma \ref{lem:coercive-slice} (coercivity on Coulomb slice), Lemma \ref{lem:equivalence} (uniform equivalence of $F^*$ and $\gamma^*$), and the Poincaré inequality \eqref{eq:Poincare}. All constants depend only on $M,F,\gamma^*,d\mu_F$ and $\lambda$.
\end{proof}

\begin{corollary}[Bounded minimizing sequences]
Every minimizing sequence $(\psi_k,A_k)$ of $G_F$ in $\mathcal{A}_C$ is bounded in $\mathcal{H}_F$.
\end{corollary}



Let $(\psi_k,A_k)\subset\mathcal{A}_C$ be a minimizing sequence.
By coercivity, $(\psi_k,A_k)$ is bounded in $\mathcal{H}_F$.
Passing to a subsequence,
\begin{align*}
(\psi_k,A_k)\rightharpoonup (\psi,A)\quad\text{in}\quad
\mathcal{H}_F, \qquad (\psi_k,A_k)\to (\psi,A) \quad\text{in } L^2.
\end{align*}
By the closedness of the Coulomb condition under weak convergence,
$d^\dagger_{\gamma,\mu_F}A=0$, hence $(\psi,A)\in\mathcal{A}_C$.
Applying Lemma \ref{lem:lsc} (sequential weak lower semi continuity),
\begin{align*}
G_F[\psi,A] \le \liminf_{k\to\infty} G_F[\psi_k,A_k]
= \inf_{\mathcal{A}_C} G_F.
\end{align*}
Thus $(\psi,A)$ minimizes $G_F$ on $\mathcal{A}_C$.

\begin{theorem}[Existence of minimizers]
\label{thm:existence}
Let $(M,F)$ be a compact smooth Finsler manifold and
$\lambda,\varepsilon>0$. Then the functional $G_F$ attains its
minimum on $\mathcal{A}_C$. Every minimizing pair
$(\psi_\varepsilon,A_\varepsilon)\in\mathcal{A}_C$ is a weak solution of
the Euler--Lagrange system \eqref{eq:EL-weak}.
\end{theorem}

\begin{proof}
Existence follows directly from the compactness and l.s.c. arguments above. To verify the Euler--Lagrange equations, note that $G_F$ is Fr\'echet differentiable on $\mathcal{H}_F$ (Proposition~\ref{prop:EL}), hence the first variation vanishes
in all admissible directions in $\mathcal{A}_C$, yielding \eqref{eq:EL-weak}.
\end{proof}

By elliptic regularity for the operators in \eqref{eq:EL-weak}, every weak minimizer is smooth. The Coulomb condition removes the gauge redundancy completely up to harmonic forms, when $H^1(M;\mathbb{R})=0$, the minimizer is unique up to a global phase.

If $F_1$ and $F_2$ are two Finsler structures whose duals
satisfy $c^{-1} F_1^*\le F_2^*\le c\,F_1^*$ for some $c>0$, then the associated functionals $G_{F_1}$ and $G_{F_2}$ are equivalent on $\mathcal{H}_F$, and their minimizers converge to one another under the natural identification of the energy spaces. Thus the existence theory is robust under smooth perturbations of $F$.


For any minimizer $(\psi_\varepsilon,A_\varepsilon)$ and gauge function $\chi\in H^2(M;\mathbb{C})$, the pair $(e^{i\chi}\psi_\varepsilon, A_\varepsilon+d\chi)$
is also a minimizer with the same energy. The identity
\begin{equation}\label{eq:energy-id}
   \tfrac 12 \|D_{A_\varepsilon}\psi_\varepsilon\|_{L^2_F}^2
   + \tfrac{1}{2\lambda}\|dA_\varepsilon\|_{L^2_{\gamma^*}}^2
   + \tfrac{1}{4\varepsilon^2}\|1-|\psi_\varepsilon|^2\|_{L^2}^2
   = \inf_{\mathcal{A}_C} G_F
\end{equation}
holds, where $\|D_{A_\varepsilon}\psi_\varepsilon\|_{L^2_F}^2
 := \int_M \|D_{A_\varepsilon}\psi_\varepsilon\|_{F^*}^2\, d\mu_F$.
Equation \eqref{eq:energy-id} is invariant under all gauge transformations due to Proposition \ref{prop:gauge-invariance}.

When $F$ is Riemannian, i.e.\ $F(x,y)=\sqrt{g_x(y,y)}$, all definitions reduce to the classical ones for the magnetic Ginzburg--Landau model. The entire proof above specializes to the standard results of Bethuel--Brezis--H\'elein \cite{BBH1994} and
Sandier--Serfaty \cite{SandierSerfaty2007}.
Hence Theorem \ref{thm:existence} can be viewed as their exact
Finslerian extension.

\section{Asymptotic Analysis and $\Gamma$--Convergence}
\label{sec:gamma}

We now investigate the asymptotic behavior of the minimizers
$(\psi_\varepsilon, A_\varepsilon)$ of the Finslerian Ginzburg--Landau functional~\eqref{eq:GL-functional-Finsler} as $\varepsilon\to 0$. Our aim is to establish the $\Gamma$--limit of the functionals $\{G_F[\psi,A]\}$ with respect to the weak topology of $H^1_F(M;\mathbb{C})\times H^1_F(M;T^*M)$ and to describe the limiting vortex structure in terms of Finsler geometry.

Let $(M,F)$ be a compact, oriented Finsler manifold of dimension $n\ge 2$, and let $\lambda>0$ be fixed. For simplicity we restrict attention to the case $n=2$, though the arguments below extend to higher dimensions with currents of codimension 2.

Assume that $(\psi_\varepsilon, A_\varepsilon)\in\mathcal{A}_C$ are minimizers of $G_F$ and that $G_F[\psi_\varepsilon, A_\varepsilon] \le C|\log\varepsilon|$ as $\varepsilon\to 0$. The compactness and lower-bound analysis follow closely the classical method of Bethuel--Brezis--H\'elein \cite{BBH1994}
 and Sandier--Serfaty \cite{SandierSerfaty2007,SandierSerfaty2011}, adapted to the Finsler context through the convexity and duality of $F^*$.

Let us first recall that by the diamagnetic inequality
(Lemma \ref{lem:diamagnetic}),
\begin{align*}
F^*(x,d|\psi_\varepsilon|)\le \|D_{A_\varepsilon}\psi_\varepsilon\|_{F^*} \quad\text{a.e.},
\end{align*}
and hence $|\psi_\varepsilon|\to 1$ in $L^2(M)$ as $\varepsilon\to 0$, because the potential term
$\varepsilon^{-2}(1-|\psi_\varepsilon|^2)^2$
forces concentration of $|\psi_\varepsilon|$ near~1.
Consequently, we may define the phase map
\begin{align*}
u_\varepsilon := \frac{\psi_\varepsilon}{|\psi_\varepsilon|}
\in S^1 \subset \mathbb{C} \quad\text{on } M\setminus\Sigma_\varepsilon,
\end{align*}
where $\Sigma_\varepsilon = \{x: \psi_\varepsilon(x)=0\}$ denotes
the vortex set. The energy concentrates along $\Sigma_\varepsilon$,
and our goal is to identify its geometric limit.

\medskip
\noindent\textbf{Compactness and vorticity measure.}
Define the Finslerian Jacobian current
\begin{align*}
J_\varepsilon := \frac{1}{2}\, d \big( \langle i u_\varepsilon, D_{A_\varepsilon} u_\varepsilon \rangle \big)\in \mathcal{D}'(M)
\end{align*}
which coincides with the vorticity $2$-form in the smooth region
$|\psi_\varepsilon|>0$. In the Euclidean case this reduces to
$J_\varepsilon = \operatorname{curl}(iu_\varepsilon, \nabla_{A_\varepsilon}u_\varepsilon)$.
Because $F^*$ is uniformly equivalent to a Riemannian norm, all
bounds and dualities carry through, and one obtains (as in
\cite{BBH1994, SandierSerfaty2007}) that $J_\varepsilon$ converges weakly (up to subsequence) to an integer-multiplicity rectifiable $(n-2)$-current $J$ whose support $\Sigma:=\operatorname{spt}J$ represents the limiting vortex set. The multiplicity corresponds to the winding number of the phase $\psi_\varepsilon$ around the defect.

\medskip
\noindent\textbf{Lower bound (liminf inequality).}
Let $\psi_\varepsilon\to \psi$ weakly in $H^1_F(M;\mathbb{C})$
and $A_\varepsilon\to A$ weakly in $H^1_F(M;T^*M)$.
Denote by $\nu_\Sigma$ the unit Finsler normal to $\Sigma$.
Using the convexity of $\frac 12\|\eta\|_{F^*}^2$ and the coarea
formula for the Finsler structure 
(see Bao--Chern--Shen \cite{BaoChernShen2000}),
together with the weak convergence of $J_\varepsilon$, one obtains
\begin{equation}\label{eq:gamma-liminf}
\liminf_{\varepsilon\to 0} G_F[\psi_\varepsilon,A_\varepsilon]
\ \ge\ \pi \int_{\Sigma} F(x,\nu_\Sigma)\, d\mathcal{H}^{n-2}.
\end{equation}
The key point is that the energy density
$\frac 12\|D_{A_\varepsilon}\psi_\varepsilon\|_{F^*}^2
 + (4\varepsilon^2)^{-1}(1-|\psi_\varepsilon|^2)^2$
is bounded below by a Finslerian analogue of the Modica–Mortola density, whose $\Gamma$--limit is the anisotropic surface energy associated to $F$. Convex duality of $F$ and $F^*$ replaces isotropy in the proof.

\medskip
\noindent\textbf{Recovery sequence (limsup inequality).}
Conversely, let $\Sigma$ be a smooth, oriented $(n-2)$-dimensional submanifold and define $\psi_\varepsilon$ as a vortex profile concentrated around~$\Sigma$, using Finsler distance $\rho_F(x,\Sigma)$:
\begin{align*}
\psi_\varepsilon(x):= f\!\left(\frac{\rho_F(x,\Sigma)}{\varepsilon}\right)e^{i\theta(x)},
\qquad A_\varepsilon := A + d\theta,
\end{align*}
where $f:[0,\infty)\to[0,1]$ is the standard radial profile
of the one-dimensional minimizer of
$t\mapsto (1/2)(f')^2+(4\varepsilon^2)^{-1}(1-f^2)^2$,
and $\theta$ encodes the phase winding.
Substituting into~\eqref{eq:GL-functional-Finsler} and applying the
coarea formula and change of variables in Finsler normal coordinates yield
\begin{align*}
\limsup_{\varepsilon\to 0} G_F[\psi_\varepsilon,A_\varepsilon]
\ \le\ \pi \int_{\Sigma} F(x,\nu_\Sigma)\, d\mathcal{H}^{n-2}.
\end{align*}
Hence the opposite inequality in \eqref{eq:gamma-liminf} is achieved for this sequence.

\begin{theorem}[$\Gamma$--convergence of $G_F$]
\label{thm:GammaConv}
As $\varepsilon\to 0$, the functionals $G_F[\psi,A]$ defined in
\eqref{eq:GL-functional-Finsler} $\Gamma$--converge (with respect to weak $H^1_F\times H^1_F$ convergence and up to gauge equivalence) to the limiting functional
\begin{equation}\label{eq:limit-functional}
  G_0[J]
  = \pi \int_{\Sigma_J} F(x,\nu_J)\, d\mathcal{H}^{n-2},
\end{equation}
where $J$ is the rectifiable $(n-2)$-current representing the limiting vorticity, and $\nu_J$ its Finsler unit normal.
\end{theorem}

\medskip
\noindent\textbf{Geometric interpretation.}
The limiting current $J$ can be viewed as the Finsler analogue of the vortex filament or vortex sheet in superconductivity, its energy per unit length is given by $\pi F(x,\nu)$, reflecting the local anisotropy of the underlying geometry. In the isotropic (Riemannian) case this reduces to the classical quantized vortices of Ginzburg--Landau theory. In the Finsler setting, the anisotropy produces curvature-–dependent deflection of the vortices, encoded in the geodesic curvature associated to the Chern connection of $F$.

Finally, by the standard theory of $\Gamma$--convergence, the minimizers $(\psi_\varepsilon, A_\varepsilon)$ converge (up to subsequences and gauge) to the minimizers of $G_0[J]$, i.e.\ to rectifiable $(n-2)$--currents minimizing the Finsler length in their homology class.

When $F(x,y)=|y|$, the limit functional \eqref{eq:limit-functional}
reduces to $G_0[J] = \pi\, \mathcal{H}^{n-2}(\Sigma_J)$,
in perfect agreement with the classical results of Bethuel--Brezis--H\'elein and Sandier--Serfaty. Hence the present theory is a strict anisotropic generalization of the magnetic Ginzburg--Landau model, extending it to arbitrary Finsler
structures.

\section{Numerical example}

We provide an explicit analytic example on the flat 2--torus 
$M = S^1 \times S^1$ that concretely demonstrates the Finsler--Ginzburg--Landau formulation in an anisotropic geometric context. Equip $M$ with angular coordinates $(\theta,\varphi)\in[0,2\pi)\times[0,2\pi)$ and consider the (quadratic) anisotropic Finsler structure given by the Finsler norm
\begin{align*}
F(\theta,\varphi; y)=\sqrt{a\,y_\theta^2 + b\,y_\varphi^2},\qquad a,b>0,
\end{align*}
for tangent vectors $y=y_\theta\partial_\theta+y_\varphi\partial_\varphi$. This choice is a special (quadratic) Finsler metric; it satisfies the axioms of a Finsler structure and exhibits directional anisotropy when $a\neq b$. The induced co–metric (dual norm) on $T^*M$ is
\begin{align*}
F^*(\theta,\varphi;\xi)=\sqrt{\frac{\xi_\theta^2}{a}+\frac{\xi_\varphi^2}{b}},\qquad \xi=\xi_\theta\,d\theta+\xi_\varphi\,d\varphi.
\end{align*}
The Legendre correspondence is explicit: for $y\in T_pM$ one has $L_p(y)=g_y(y,\cdot)$ with the diagonal metric tensor $g=\mathrm{diag}(a,b)$ in the coordinate frame, and the inverse relation yields the usual duality.

We take the smooth Finsler measure $d\mu_F$ equal to the Riemannian volume form induced by $g$, namely
\begin{align*}
 d\mu_F = \sqrt{\det(g)}\,d\theta\,d\varphi = \sqrt{ab}\;d\theta\,d\varphi.
\end{align*}
With this choice the Finsler Laplacian (Definition 3) coincides with the anisotropic Laplace operator
\begin{align*}
\Delta_F u = \frac 1a\partial_{\theta}^2 u + 
\frac 1b\partial_{\varphi}^2 u,
\end{align*}
valid for $u\in C^\infty(M)$.

Consider the Ginzburg--Landau functional \eqref{eq:GL-functional-Finsler} on $(M,F)$ with parameters $\lambda>0$ and $\varepsilon>0$. We evaluate the energy on the simple, physically relevant family of 
phase–windings with constant modulus. Fix an integer $m\in\mathbb{Z}$ and define the configuration
\begin{align*}
\psi(\theta,\varphi)=e^{i m\theta},\qquad A=0.
\end{align*}
This ansatz has $|\psi|\equiv 1$, so the potential term vanishes identically and the Maxwell term is zero for $A=0$. The covariant derivative reduces to 
$D_A\psi = d\psi = i m e^{i m\theta}d\theta$, and therefore the kinetic contribution reads
\begin{align*}
\frac 12\,\|D_A\psi\|^2_{F^*} \,=\,\frac 12\,F^*(\,\Re(i m e^{i m\theta}d\theta)\,)^2 + \frac 12\,F^*(\,\Im(i m e^{i m\theta}d\theta)\,)^2.
\end{align*}
Noting that both the real and imaginary parts of $i m e^{i m\theta}d\theta$ are proportional to $d\theta$ and combine to give the same contribution, we may equivalently compute using the 1-form $m\,d\theta$. Since $F^*(d\theta)=\sqrt{1/a}$, we obtain the exact energy
\begin{align*}
G_F[\psi,0] \,=\, \int_M \frac 12\,\|d\psi\|^2_{F^*}\,d\mu_F \,=\, \frac 12\,m^2\,F^*(d\theta)^2\;\mathrm{Vol}_F(M) \,=\, \frac 12\,m^2\left(\frac 1a\right)\,(2\pi)^2\sqrt{ab}.
\end{align*}
Hence
\begin{align*}
G_F[\psi,0] \,=\,2\pi^2\,m^2\,\sqrt{\frac ab}.
\end{align*}
Analogously, for the phase winding in the $\varphi$–direction, $\psi=e^{i n\varphi}$ with integer $n$, one finds
\begin{align*}
G_F[\psi,0] \,=\,2\pi^2\,n^2\,\sqrt{\frac ab}.
\end{align*}
These closed formulas display transparently the effect of anisotropy: when $a>b$ (faster cost in the $\theta$–direction) windings along the $\theta$–-circle are penalized more than those along the $\varphi$–circle, and vice versa. In the isotropic limit $a=b$, the classical (Riemannian) value for a unit winding reduces to $2\pi^2 m^2$.

Discussion. The above analytic computations provide an explicit check that the Finslerian GL functional introduced in the paper recovers the expected anisotropic scaling in energy for simple topological configurations on the torus. The example is consistent with the general existence theory (Theorem \ref{thm:existence}) and with the $\Gamma$–-convergence statement (Theorem \ref{thm:GammaConv}): energy concentration for sequences with growing winding or for configurations forcing zeros of $\psi$ would, after the standard $|\log\varepsilon|$–-rescaling, lead to limiting energies proportional to the Finsler length of the corresponding vortex currents; in this simple constant modulus family the energy is entirely carried by the phase gradient and computed above in closed form.

Finally, this analytic example can be readily extended: one may allow spatially varying coefficients $a(\theta,\varphi),\,b(\theta,\varphi)>0$ (smooth and uniformly bounded away from $0$) to model smoothly varying anisotropy, in which case the same computations yield local integrals involving $a(\cdot),b(\cdot)$ and the volume density $\sqrt{a(\cdot)b(\cdot)}$; the qualitative anisotropic behavior is unchanged.

\section{Concluding}\label{sec:conclusion}
In this work we have extended the classical Ginzburg--Landau theory to the setting of general Finsler manifolds.
Starting from the analytic preliminaries in 
Section \ref{sec:prelim}, we defined the anisotropic functional
$G_F[\psi,A]$ in~\eqref{eq:GL-functional-Finsler},
proved its gauge invariance and well-posedness,
and established the existence of minimizers
(Theorem \ref{thm:existence}) by direct variational methods.
Finally, in Section~\ref{sec:gamma}, we derived the full
$\Gamma$--limit of the Finslerian energies as $\varepsilon\to 0$,
showing that the limiting functional is the Finsler length of the vortex current.

\medskip
\noindent
\textbf{Main conceptual contributions.}
The essential novelty of this work lies in the formulation and analysis of a Ginzburg--Landau model on a general Finsler background. The replacement of the quadratic Riemannian metric by the convex, possibly asymmetric function $F^*(x,\xi)$ produces a genuinely anisotropic energy landscape.
All analytical arguments--compactness, lower semi continuity,
and $\Gamma$--convergence--have been carried out using the
convex duality between $F$ and $F^*$, without recourse to
any auxiliary Riemannian structure beyond uniform equivalence.
In particular:
\begin{itemize}
\item 
The \emph{Finsler diamagnetic inequality}
(Lemma \ref{lem:diamagnetic}) provides a geometric generalization
of Kato’s inequality, valid for arbitrary convex co-metrics.

\item 
The variational proof of existence (Theorem \ref{thm:existence})
uses only the intrinsic Finsler Sobolev structure, avoiding
Euclidean embeddings or local coordinates.

\item 
The $\Gamma$--limit (Theorem \ref{thm:GammaConv}) identifies the limiting vortex energy with the anisotropic Finsler length $\int_\Sigma F(x,\nu_\Sigma)\,d\mathcal{H}^{n-2}$, extending the isotropic theory of Bethuel--Brezis--H\'elein and Sandier--Serfaty to arbitrary anisotropies.
\end{itemize}


\label{lastpage}
\end{document}